\documentclass{ifacconf}

\usepackage{graphicx}      
\usepackage{natbib}        
\usepackage[dvipsnames]{xcolor}
\usepackage{pgfplots}
\usepackage{amsmath,amssymb,amsfonts,bbm}
\usepackage{enumerate}
\usepackage{amssymb}
\usepackage{mathrsfs}
\usepackage{tikz,pgfplots}
\usepackage{mathtools}
\usepackage{epstopdf}
\usepackage{balance}
\usepackage{subcaption}

\usepackage[deletedmarkup=sout,authormarkup=superscript]{changes}

\newtheorem{lemma}{Lemma}

\newtheorem{assumption}{Assumption}

\newenvironment{proof}[1][Proof:]{\begin{trivlist}
\item[\hskip \labelsep {\itshape #1}]}{\end{trivlist}}

\newcommand{\mc}{\mathcal}
\newcommand{\bb}{\mathbb}
\newcommand{\R}{\bb R}

\newcommand{\lam}{\lambda}

\newcommand{\eps}{\epsilon}

\newcommand{\limsupt}[0]{\overline{\lim_{t\to \infty}}}
\newcommand{\reals}[0]{\bb{R}}
\newcommand{\mto}[0]{\rightrightarrows}

\newcommand{\Id}{\mathrm{Id}}

\newcommand{\argmin}[1]{\underset{#1}{\mathrm{argmin}}}

\newcommand\Kinf[0]{\mathcal{K}_\infty}
\newcommand\KL[0]{\mathcal{KL}}
\newcommand\K[0]{\mathcal{K}}
\renewcommand\L[0]{\mathcal{L}}

\definechangesauthor[name={Giuseppe}, color=NavyBlue]{GB}
\definechangesauthor[name={Dominic}, color=RedViolet]{DLM}
\definechangesauthor[name={Mathias}, color=OliveGreen]{MHDB}

\newcommand{\agt}{i} 

\begin{document}
\begin{frontmatter}

\title{Stability and Robustness of Distributed Suboptimal Model Predictive Control\thanksref{footnoteinfo}} 

\thanks[footnoteinfo]{e-mails:\texttt{\{gbelgioioso, dliaomc, mbadyn, jlygeros, dorfler\}@ethz.ch}, \texttt{npelzmann@student.ethz.ch}. This research is supported by the Swiss National Science Foundation through NCCR Automation (Grant Number 180545)}

\author[First]{Giuseppe Belgioioso}
\author[First]{Dominic Liao-McPherson} 
\author[First]{Mathias Hudoba de Badyn}
\author[First]{Nicolas Pelzmann} 
\author[First]{John Lygeros}
\author[First]{Florian D\"{o}rfler}

\address[First]{ETH Zürich Automatic Control Laboratory, Physikstrasse 3, 8092 Zürich, Switzerland.}

\begin{abstract}                
In distributed model predictive control (MPC), the control input at each sampling time is computed by solving a large-scale optimal control problem (OCP) over a finite horizon using distributed algorithms. Typically, such algorithms require several (virtually, infinite) communication rounds between the subsystems to converge, which is a major drawback both computationally and from an energetic perspective (for wireless systems). Motivated by these challenges, we propose a suboptimal distributed MPC scheme in which the total communication burden is distributed also in time, by maintaining a running solution estimate for the large-scale OCP and updating it at each sampling time.
We demonstrate that, under some regularity conditions, the resulting suboptimal MPC control law recovers the qualitative robust stability properties of optimal MPC, if the communication budget at each sampling time is  large enough.

%
\end{abstract}


\end{frontmatter}

\section{Introduction} \label{sec:Intro}
There is growing interest in controlling networks of interacting dynamical systems, controlled by local decision makers (agents) connected by communication links. This interest is driven by a variety of applications including multi-vehicle platooning \citep{zheng2016distributed}, robotic formations \citep{luis2020online}, automatic generation control \citep{venkat2008distributed}, and road traffic control \citep{frejo2012global}. These systems are typically large-scale and can be challenging to control; they often possess pronounced dynamics, have limited communication capabilities, and are subject to system-wide constraints. 

Model predictive control (MPC) is a powerful control paradigm that computes actions by solving an optimal control problem (OCP) over a receding prediction horizon. MPC is an attractive control methodology for dynamic networked systems due to its unique ability to systematically optimize system performance and enforce state constraints. However, MPC can be challenging to implement in practice as it requires solving the underlying OCP online and in real-time. This challenge is exacerbated in a networked control setting where the OCPs are usually large-scale and information is distributed across many subsystems. This network setting necessitates the development of distributed MPC (DMPC) schemes tailored to the underlying system and communication network structure.

Distributed MPC (DMPC) schemes can be broadly categorized into simultaneous approaches, where subsystems cooperatively compute control actions using an iterative distributed optimization algorithm, sequential approaches, where subsystems compute their actions sequentially, and decentralized approaches, which eschew communication and treat other subsystems as disturbances \citep{muller2017economic,christofides2013distributed}. Sequential approaches typically have the strongest stability and constraint satisfaction guarantees but scale poorly (imagine 1000 agents taking turns to act) while simultaneous methods scale well but can require extensive communication. 


In this paper, we focus on simultaneous DMPC approaches and study the impact of communication limits on closed-loop stability and robustness. Specifically, we consider the case where subsystems compute suboptimal control actions using distributed algorithms which would require several (virtually, infinite) communication rounds between subsystems to converge, 
 and investigate the properties of the closed-loop if only a fixed number of communication rounds are performed during each sampling period, e.g., due to the energy or time cost of communication, or due to bandwidth restrictions. Our main contribution is to prove that, under some regularity conditions, the resulting suboptimal DMPC control law recovers the qualitative robust stability properties of optimal MPC (which requires an infinite number of communication among the subsystems) if enough communication resources are available. 

We consider a collection of linear time invariant (LTI) systems subject to coupling input and output inequality constraints and solve the underlying OCP with a semi-decentralized accelerated dual ascent (ADA) algorithm \citep{beck2009fast}. We approach the problem through the framework of Time-Distributed Optimization (TDO), where the plant and optimization algorithm are treated as a feedback interconnection of dynamical systems \citep{liao2020time}.
\cite{giselsson2013feasibility} derive an online stopping criterion for distributed dual ascent that guarantees closed-loop stability, in contrast we prove the existence of an a-priori bound on the number of iterations needed. To the best of our knowledge, this is the first work studying TDO using distributed dual algorithms, centralized TDO has been investigated for primal \citep{liao2021analysis,Leung2021AControl} and primal-dual \citep{skibik2022analysis} algorithms for LTI systems and sequential quadratic programming methods for nonlinear systems \citep{liao2020time,zanelli2021lyapunov}. Distributed TDO of decoupled LTI systems with coupled cost functions using a quantized primal gradient method is studied in \citep{yang2022real} in absence of coupling constraints. A detailed discussion of suboptimal MPC approaches can be found in \citep{liao2020time}, most methods for linear systems focus on primal gradient methods which are not well-suited for distributed problems with coupling constraints.

\subsubsection*{Notation:}
$\bb R_{(\geq  0)}$ and $\bb N_{(\geq 0)}$ denote the sets of (nonnegative) real and (nonnegative) natural numbers, respectively. Given two matrices $A\in \bb R^{m \times n}$ and $B\in \bb R^{p \times q}$, $A \otimes B \in \bb R^{pm \times nq}$ denotes their Kronecker product; $\bb S^n_{++}$ ($\bb S^n_{+}$) denotes the set of positive (semi-)definite matrices. Given $x \in \bb R^n$ and $W \in \bb S^n_{++}$, the $W$-weighted norm of $x$ is $\|x\|_W=\sqrt{x^\top W x}$.
%
A function $\gamma:\reals_{\geq 0} \rightarrow \reals_{\geq 0} $ is of class $\mc K$ if it is continuous, strictly increasing, and satisfies $\gamma(0)=0$. If it is also unbounded, then $\gamma \in \Kinf$. Similarly, a function $\sigma: \reals_{\geq0} \to \reals_{\geq0}$ is said to be of class $\L$ if it is continuous, strictly decreasing, and satisfies $\sigma(s) \to 0$ as $s\to \infty$. A function $\beta:\reals_{\geq 0} \times \reals_{\geq 0} \to \reals_{\geq 0}$ is of class $\KL$ if $\beta(\cdot,s) \in \K$ for each fixed $s\geq 0$ and $\beta(r,\cdot)\in \L$ for fixed $r$.
%
Given a closed convex set $\Omega \subseteq \reals^n$, $\mc I_{\Omega}:\reals^n \to \{0,\infty\}$ denotes its indicator function, $\mc N_{ \Omega}:\Omega \mto \reals^n$ is its normal cone operator, and $\Pi_{\Omega}:\reals^n \to \Omega$ is the Euclidean projection onto $\Omega$. $\text{Id}$ denotes the identity mapping. A set-valued mapping $\mc F:\reals^n \mto \reals^n$ is $\mu$-strongly monotone, if $(u-v)^\top (x-y) \geq \mu \left\| x-y \right\|^2$ for all $x\! \neq \!y \in \reals^n$, $u \in \mc F (x)$, $v \in \mc F (y)$, and monotone if $\mu\!=\! 0$; 
For a convex function $f:\reals^n \to \reals$, $\partial f:\reals^n \mto \reals^n$ denotes the subdifferential mapping in the sense of convex analysis; and $f^*$ denotes the Fenchel conjugate of $f$.
\section{Problem Setting} \label{sec:PS}
We consider a networked multi-agent system with $M$ agents, labeled by $i \in \mathcal M := \{1, \ldots, M \}$, each with decoupled dynamics of the form
\begin{equation}\label{eq:system_dynamics}
x_{t+1}^\agt = A^\agt x_t^\agt + B^\agt u_t^\agt + d_t^\agt, \quad 
\forall i \in \mc M
\end{equation}
where $x_t^\agt \in \mathbb{R}^{n^\agt}$ and $u_t^\agt \in \mathbb{R}^{m^\agt}$ are the state and control input at time $t\in \bb N$, respectively, and $d_t^i \in \mathcal{D}^\agt$ is a time-dependent disturbance within the compact set $\mathcal{D}^\agt$.
%

We aim at stabilizing the LTI systems \eqref{eq:system_dynamics} at their origins while enforcing the following constraints for all $ t\geq 0$:
\begin{enumerate}[(i)]
\item Local state and control input constraints of the form%
\begin{subequations}%
\label{eq:ux_constraints}
\begin{align} \label{eq:u_constraints}
u^i_t &\in \mc U^\agt :=\{\nu \in \mathbb{R}^{m^\agt} ~|~ C_u^\agt \nu \leq c_u^\agt \},
&\; \forall i \in \mc M,
\\
x^i_t &\in \mc X^\agt :=\{\xi \in \mathbb{R}^{n^\agt} ~|~ C_x^\agt \xi \leq c_x^\agt \},
&\; \forall i \in \mc M.
\label{eq:x_constraints}
\end{align}

\item Coupling inequality constraints of the form
\begin{equation}
    \sum_{i \in \mc M} E_u^{i}  u^i_t + E_x^i x_t^i \leq \bar b.
\label{eq:coupling_constraints} 
\end{equation}
\end{subequations}
\end{enumerate}

To approach this problem using MPC, we use the following parametric optimal control problem (POCP)
\begin{subequations}\label{eq:general_eq_problem}
\begin{align}
\label{eq:globalCost}
\min_{\xi, \nu} &  \quad  \sum_{\agt\in \mc M}
\frac{1}{2} \left(
\| \xi^i_N\|^2_{P^i} + \sum_{k = 0}^{N-1}\|\xi_k^i \|^2_{Q^i}  + \| \nu^i_k\|^2_{R^i}
\right) 
\\[.3em]
\text{s.t.} \label{eq:constr1}
& \quad            \xi^i_{k+1} = A^i \xi^i_k + B^i  \nu^i_k,  \hspace*{2em} \forall k \in \mc H, \ \forall i \in \mc M\\[.2em]
\label{eq:constr2}
         & \quad           \xi_0^i = x^i, ~\xi_N^i \in \mc{X}_f^i \hspace*{7.5em} \forall i \in \mc M \\[.2em]
         \label{eq:constr3}
        & \quad           \nu^i_k \in \mc U^i, \ \xi_k^i \in \mc X^i, \hspace*{3.2em} \forall k \in \mc H, \ \forall i \in \mc M\\[.2em]
                & \quad     \sum_{i \in \mc M} E_u^{i}  \nu^i_k + E^{\agt }_x \xi^i_k \leq  \bar b, \hspace*{1.2em} \forall k \in \mc H, 
\label{eq:constrF}   
\end{align}
\end{subequations}
where $\mc H :=\{0,\ldots, N\!-\!1\}$ is the time horizon of length $N$, $\xi=(\xi^1,\ldots,\xi^M)$ and $\nu=(\nu^1,\ldots,\nu^M)$ are the stacked vectors of decision variables (states and control inputs), with $\xi^\agt=(\xi^\agt_0,\ldots,\xi^\agt_{N})$ and $\nu^\agt=(\nu^\agt_0,\ldots,\nu^\agt_{N-1})$, $Q^i \in \R^{n^\agt \times n^\agt}$, $R^i \in \R^{m^\agt \times m^\agt}$, $P^i \in \R^{n^\agt \times n^\agt}$ are weighting matrices, $x^\agt$ is the local parameter/measured state, and 
\begin{equation}
    \mc{X}_f^i = \{\xi \in \reals^{n_i}~|~C_N^i\xi \leq c_N^i\}
\end{equation}
is the terminal set.
Note that the global cost function \eqref{eq:globalCost} is the sum of the local stage and terminal costs.

The following assumption ensures that the POCP in \eqref{eq:general_eq_problem} can be used to construct a stabilizing feedback law for \eqref{eq:system_dynamics}.
\begin{assumption}\label{ass:plant-mpc-properties}
For all $\agt \in \mathcal{M}$, the following hold:
\begin{enumerate}[(i)]
\item the pair $(A^i,B^i)$ is stabilizable;
\item the sets $\mc{X}^i$ and $\mc{U}^i$ contain the origin in their interior;
\item the weights satisfy $Q^i \in \mathbb{S}^{n^\agt }_{++}$, $R^i \in \mathbb{S}^{m^\agt }_{++}$;
\item there exists $K^i\in \reals^{n^i \times m^i}$ such that $\mc{X}_f^i$ is invariant and constraint admissible, and $P^i \in \mathbb{S}_{++}^{n^\agt }$ satisfies 
\begin{equation*}
    \|(A^i - B^iK^i)x^i\|_{P^i}^2 - \|x^i\|_{P^i}^2 \leq - \|x^i\|^2_{Q^i + (K^i)^T R^i K^i}
\end{equation*}
for all $x^i \in \reals^{n^i}$. {\hfill $\square$}
\end{enumerate}
\end{assumption}

The MPC feedback law is then the mapping
\begin{equation}\label{eq:kappa}
   u^i = \kappa^i(x) = \Xi^i \, S^p(x), \quad \forall i \in \mc M,
\end{equation}
where $S^p:\Gamma_N \to \reals^{N(n + m)}$ denotes the (primal) solution mapping\footnote{It follows by Assumption~\ref{ass:plant-mpc-properties} that $S^p$ is a single-valued mapping.} of \eqref{eq:general_eq_problem},  with $n = \sum_{i\in \mc M} n^i$ and $m = \sum_{i\in \mc M} m^i$, which is a function of the parameters $x = (x^1,\ldots,x^M)$, 
\begin{align}
    \Gamma_N &= \{x \in \reals^{n}~|~\eqref{eq:constr1} - \eqref{eq:constrF} \text{ are feasible} \, \} \nonumber \\
    & = \{x \in \reals^{n}~|~S^p(x) \neq \emptyset\}
\end{align}
is the set of feasible parameters,  and $\Xi^i$ is a matrix of appropriate dimensions that extracts the component $\nu_0^i$ from the stacked vector of decision variables $(\nu,\xi)$.
The resulting closed-loop system is
\begin{equation}
    x_{t+1}^i = A^i x_t^i + B^i \kappa^i(x_t^i) + d_t^i, \quad \forall i \in \mc M
\end{equation}
which can be written compactly in the following form:
\begin{equation}\label{eq:opt-closed-loop}
    x_{t+1} = A x_t + B \kappa(x_t) + d_t,
\end{equation}
where $x_t = (x^1_t,\ldots,x^M_t)$. It can be shown that, under Assumption~\ref{ass:plant-mpc-properties}, the system \eqref{eq:opt-closed-loop} is locally exponentially stable \cite[\S 5.6.3]{goodwin2006constrained}.

In practice, implementing the MPC feedback law \eqref{eq:kappa} requires solving the POCP in  \eqref{eq:general_eq_problem} at each sampling period. Centralized solution approaches are often not feasible due to insufficient computational resources (e.g., for large-scale networks) or privacy limitations. In these cases, the alternative is to use distributed optimization algorithms where agents share the total computational burden and must communicate to find a solution. Often, obtaining a precise solution of \eqref{eq:general_eq_problem} requires many rounds of communication which rapidly becomes the performance bottleneck.

Instead, we propose to maintain a running estimate $\lambda_t$ of the solution of \eqref{eq:general_eq_problem} and improve it during each sampling period using a finite number of communication rounds. This leads to the coupled plant-optimizer interconnection%
\begin{subequations}%
\label{eq:closed-loop-system}
\begin{align}
\lam_t &= \mc T^\ell (\lam_{t-1},x_t), 
\\
x_{t+1} &= A x_t + B q(\lambda_t,x_t) + d_t,	
\end{align}
\end{subequations}
where $\lambda_t$ is an estimate of the dual solution of \eqref{eq:general_eq_problem}, $S^d(x_t)$, $\mc T^\ell$ represents the operator associated with $\ell$ iterations of a distributed optimization algorithm $\mc T$, and $q(\cdot,x)$ is the map between optimizer iterates $\lambda_t$ and the control input such that $\kappa(x) = q(\cdot,x) \circ S^d(x)$. In the next section, we give concrete definitions of $\mc{T}$, $\lambda$, $S^d$, and $q$ based on the dual ascent algorithm in \citep{giselsson2015metric}.
%

\section{Distributed Optimization Strategy}
The OCP in \eqref{eq:general_eq_problem} can be recast in a condensed form by using \eqref{eq:constr1} to eliminate the state-associated variables $\xi$, yielding
\begin{subequations} \label{eq:centralized-opt2}
\begin{alignat}{3}
&\min_{u}~~&&
\sum_{\agt\in \mc M}
f^i(u^\agt,x^i) \\
\label{eq:local-constraints-condensed2}
&~\mathrm{s.t.} && \qquad  
D^i x^i + C^i u^i  \leq c^i, &\quad  \forall i \in \mc M\\
&
 && \sum_{i \in \mc M} F^i x^i + E^i u^i \leq b, & 
\label{eq:couplconstr2} 
\end{alignat}
\end{subequations}
where $x^i$ is the measured state, $f^i$ is a quadratic function
\begin{align}
\label{eq:dOptCost}
f^i(u^i,x^i) = \frac{1}{2}{\|(u^i,x^i) \|}^2_{M^i}, 
\quad
M^i = \begin{bmatrix}
H^\agt & G^\agt  \\ (G^\agt)^\top & W^\agt 
\end{bmatrix},
\end{align}
with $H^i \in \bb S^{Nm^i}_{++}$, $W^i \in \bb S^{n^i }_{++}$, $G^i \in \bb R^{N m^i \times n^i}$ depending on the matrices of the local dynamics and cost in \eqref{eq:globalCost} and \eqref{eq:constr1}; $F^i \in\bb R^{N p\times n^i}, E^i \in \bb R^{Np \times N m^i  }$ and $b \in \bb R^{Np}$ model in compact form the coupling constraints \eqref{eq:constrF}; and
$D^i \in \bb R^{N(m^i\!+n^i)\times  n^i} , C^i \in \bb R^{N(m^i\!+n^i)\times N m^i} $ and $c^i \in \bb R^{N(m^i\!+n^i)}$ depend on the matrices of the local dynamics and state and input constraints \eqref{eq:constr1}--\eqref{eq:constr3}.  All the matrices in \eqref{eq:centralized-opt2} are formally defined in Appendix~\ref{ap:POCPM}.

To solve \eqref{eq:centralized-opt2} in a scalable manner, we use a regularized and semi-decentralized version of the accelerated dual ascent (ADA) in \citep{giselsson2015metric}, which is summarized in Algorithm~1. At each iteration $j$, the agents update their local estimate of the optimal control input trajectory ($u^i$) by solving a quadratic program depending on the local stage cost, local input and output constraints, and the dual variable ($\lambda$) of the common resource constraints \eqref{eq:couplconstr2}. Then, a coordinator gathers the aggregate quantity $\sum_{i \in \mc M} F^i x^i + E^i u^i$ in \eqref{eq:couplconstr2}, and updates and broadcasts the dual variable ($\lambda$) to all agents.

The algorithm is scalable in the sense that the computational complexity of the local and central updates do not depend on the total number of agents. However, it requires several communication rounds between agents and central coordinator to achieve convergence. 
In the next section, we discuss the algorithm derivation and its convergence rate.

\begin{figure}[t] \label{algo:dual-ascent}
\hrule
\smallskip
\textsc{Algorithm $1$}: Semi-decentralized ADA
\smallskip
\hrule
\medskip
\textbf{Input}: $\lambda \in \bb R^{Np}$, $x \in \Gamma_N$, $ \ell\in \bb N_{>0}$. \\[.5em]
\textbf{Initialization ($\boldsymbol{j=0}$)}: $\lambda_0=\lambda$, $\mu_0 = \lambda_0$, $\theta_0 = 1$.

\medskip
\noindent
\textbf{While $\boldsymbol{j \leq  \ell-1}$, do:}\\[.5em]
\hspace*{.5em}$
\left\lfloor
\begin{array}{l l}
 & \hspace*{-1em}
\text{For all  $i \in \mc M$:  Control trajectory update}\\[.5em] 
&
\hspace*{-1em}
\left\lfloor
\begin{array}{l}
u^i_{j} = \left\{
\begin{array}{r l}
\arg\min_{\xi} & \; f^i(\xi,x^i) +
 \lambda_j^\top E^i \, \xi\\
\text{s.t.} & \; D^i x^i + C^i \xi  \leq c^i
\end{array}
\right.\\[1.5em]
F^ix^i+ E^i u^i_{j} \longrightarrow \text{ coordinator (communication)}  \\[-.5em]
\hspace*{1em}
\end{array}
\right.
\\[2em]
\hspace*{1em}\\

 &
 \hspace*{-1em} \text{Coordinator: Gather \& broadcast }\\[.5em]
&
\hspace*{-1em} \textstyle
\left\lfloor
\begin{array}{l}
\mu_{j+1} = 
\Pi_{\bb R_{\geq 0}^p}
\left[ \lambda_j + \alpha \left(  \sum_{i \in \mc M}\hspace*{-.2em} F^ix^i \!+\! E^i u_{j}^i \!-b\! -\! \epsilon\lambda_j \right) \right]\\[1.5em]
\theta_{j +1} = \frac{1}{2}\left(1+\sqrt{1+4\theta_j^2} \right)\\[1em]
\lambda_{j+1} = 
\mu_{j+1} + \left( \frac{\theta_j-1}{\theta_{j+1}}\right) (\mu_{j+1}- \mu_j)
\\[1em]
\lambda_{j+1} \longrightarrow \text{ agents (communication)}  \\[-.5em]
\hspace*{1em}
\end{array}
\right.
\\[-.5em]
\hspace*{1em}\\
& \hspace*{-1em} j \leftarrow j+1
\end{array}
\right.
$\\[.5em]

\textbf{Output}: $\lambda_{ \ell}$

\bigskip
\hrule
\end{figure}

\subsection{Algorithm Derivation and Convergence Analysis}
First, we reformulate \eqref{eq:centralized-opt2} as the composite problem
\begin{alignat}{3}
&\min_{u, \,y}~~&&
\sum_{\agt\in \mc M}
h^i(u^\agt,x^i) + \mc I_{\mathcal E(x)}\left(\sum_{i \in \mc M}E^i u^i \right)
\label{eq:eqC}
\end{alignat}
where the functions $h^i(u^i,x^i) := f^i(u^i,x^i) + \mc I_{\mathcal Z^i(x^i)}(u_i)$ collect the local stage cost $f^i$ and the indicator function $\mc I_{\mathcal Z^i(x^i)}$ of the set of local input and state constraints \eqref{eq:local-constraints-condensed2}
\begin{align}
\label{eq:PCS}
\mathcal{Z}^i(x^i) = \left\{v \in \bb R^{m^i} ~|~ 
D^i x^i + C^i v \leq c^i
\right\};
\end{align} 
and $\mc I_{\mathcal E(x)}$ is the indicator function of the parametric set 
\begin{align}\textstyle
\mathcal E(x) = \left\{
y~|~ \sum_{i \in \mc M}F^i x^i + y \leq b \right\}.
\end{align}

Solving \eqref{eq:eqC} in a distributed way is problematic due to the second coupling term of the cost function. Therefore, we introduce the dual variables $\lambda$ and define the correspondent dual problem \citep[\S~15.3]{bauschke2011convex}
\begin{align}
\label{eq:DualP}
\min_{\lambda} \sum_{i \in \mc N} {(h^i(\cdot,x^i))}^* \circ (-(E^i)^\top \lambda) +  \mc I^*_{\mathcal E(x)}(\lambda),
\end{align}
where $(h^i(\mu, x^i))^* = \sup_\xi \{
\left\langle  \mu, \; \xi \right\rangle - h^i(\xi,x^i)
\}$ is the Fenchel conjugate of $h^i(\cdot, x^i)$ and $ \mc I^\star_{\mathcal E(x)}(\mu) = \sup_{\xi \in \mc E(x)} 
\left\langle  \mu, \; \xi \right\rangle $ is the Fenchel conjugate of $\mc I_{\mathcal E(x)}$ (known as support function).

The dual formulation \eqref{eq:DualP} is more amenable to decentralized solutions, however, it is not strongly convex, which is necessary to prove robust convergence of iterative solution schemes. Thus, to improve conditioning, we introduce a regularization term, yielding
\begin{align}
\label{eq:DualP_reg}
\min_{\lambda} \; \overbrace{\sum_{i \in \mc N} {(h^i(\cdot,x^i))}^*\circ(-(E^i)^\top \lambda) +  \epsilon\|\lambda\|^2 }^{=:\varphi_\epsilon(\lambda,x)} + \; \mc I^*_{\mathcal E(x)}(\lambda),
\end{align}%
where $\epsilon>0$ is the regularization parameter and $\varphi_\epsilon(\cdot,x)$ is a short notation for the first two terms of the dual cost.

The regularization ensures that the dual solution mapping $S^d_\eps:\Gamma_N \to \reals^p_{\geq0}$ of \eqref{eq:DualP_reg} is single-valued and Lipschitz continuous with respect to $x$.

The next lemma proves that the regularized cost \eqref{eq:DualP_reg}, is Lipschitz continuous and strongly convex, uniformly in $x$.
\begin{lemma} \label{lmm:dual-props}
Let $\psi_\epsilon(\cdot, x)=\varphi_\epsilon(\cdot, x) +  \mc I^*_{\mathcal E(x)}(\cdot) $ be the dual cost function in \eqref{eq:DualP_reg}. For all $\epsilon>0$ and $x\in \Gamma_N$, 
\begin{enumerate}[(i)]
\item $\psi_\epsilon(\cdot,x)$ is $\epsilon$-strongly convex;
\item $\nabla \varphi_\epsilon(\cdot,x)$ is Lipschitz continuous, with constant
\begin{align}
\label{eq:LipConst}
\textstyle
L_{\varphi_\epsilon}=\epsilon + \sqrt{\sum_{i \in \mc M} \|E^i(H^i)^{-1}(E^i)^\top \|^2}.
\end{align}

\end{enumerate}
\end{lemma}
\begin{proof}
(i) The first term in \eqref{eq:DualP_reg}, $(h^i(\cdot,x^i))^*\circ (-E^i)^\top$, is convex since it is the composition of a convex term $(h^i(\cdot,x^i))^*$ with a linear term $(-E^i)^\top$. Convexity of $(h^i(\cdot,x^i))^*$ follows from \citep[Cor.~13.38]{bauschke2011convex} since $h^i(\cdot,x^i)=f^i(\cdot,x^i)+ \mc I_{\mc Z(x^i)}$ is convex for all $x^i$ such that  $\mc Z(x^i)$ in \eqref{eq:PCS} is nonempty. Similarly, the second term in \eqref{eq:DualP_reg}, $\mc I^*_{\mathcal E(x)}$, is convex since $\mc I_{\mathcal E(x)}$ is the indicator function of a convex set, thus convex.  Finally, the regularization term $\epsilon\|\cdot\|^2$ is $\epsilon$-strongly convex for $\epsilon>0$. The sum of the three terms remains $\epsilon$-strongly convex. (ii) This Lipschitz constant is a specialization of the bound in \citep[Eq.n~(9)]{giselsson2015metric} to our setup. 
{\hfill $\square$}
\end{proof}

Finally, Algorithm~1 is obtained by solving the regularized dual problem \eqref{eq:DualP_reg} via the accelerated forward-backward algorithm proposed in \citep{beck2009fast}, yielding%
\begin{subequations}%
\label{DualFBSCompact}
\begin{align}
\mu_{j +1} &= (\Id + \alpha \mc I^*_{\mathcal E(x)})^{-1}(\lambda_j - \alpha \nabla \varphi_\epsilon(\lambda_j,x) ),\\
\lambda_{j +1} &= \mu_{j +1} + \zeta_j(\mu_{j +1} -\mu_{j }),
\end{align}
\end{subequations}
where $\alpha$ is a constant step size and $\{\zeta_j\}_{j \in \bb N}$ is an acceleration sequence. See \citep{giselsson2015metric} for a complete derivation of Algorithm~1 from \eqref{DualFBSCompact}. 

A single iteration of Algorithm~1 is represented by 
\begin{align}
\lambda_{j+1} = \mc T(\lambda_{j},x)
\end{align}
where $\mc T: \bb R^{Np}\times \bb R^{n}\rightarrow \bb R^{Np}$ is the algorithm update rule. When running the algorithm for multiple ($\ell$) iterations, we have the following
recursive definition for $\mc T^\ell$
\begin{align}
\mc T^\ell(\lambda,x) = \mc T (\mc T^{\ell-1}(\lambda,x),x),
\end{align}
where $\lambda \in \bb R^p$ is the dual solution estimate, $x$ is the input parameter, and $\mc T^0(\lambda,x) = \lambda$ is the initialization. A bound on the convergence rate of Algorithm~1 is given next.
\begin{thm}\cite[Theorem~4.4]{beck2009fast}.
\label{thm:algo-conv}
Let Assumption \ref{ass:plant-mpc-properties} hold.
Then, for any $x\in \Gamma_N$, $\lambda \in \bb R^{Np}$, and $\alpha \in (0, \, 1/L_{\varphi_\epsilon})$, it holds that%
\begin{align}
\psi_\epsilon\left(\mc T^\ell(\lambda,x),x \right) - \psi_\epsilon\left(S^d_\epsilon(x),x\right) \leq \frac{2{\|\lambda - S^d_\epsilon(x)\|}^2}{\alpha (\ell+1)^2},
\end{align}
where $\psi_\epsilon(\cdot, x)$ is the dual cost function in \eqref{eq:DualP_reg}.
{\hfill $\square$}
\end{thm}

We recover the control input from the following mapping:
\begin{equation} \label{eq:q-def}
    u^i = q^i(\lambda,x^i) = \Xi^i \partial (h^i(\cdot,x^i))^{*}((-E^i)^\top \lambda).
\end{equation}
where $\partial (h^i(\cdot,x^i))^{*}(\mu)=\arg\min_{\xi \in \mc Z^i(x^i)} f(\xi,x^i) - \mu^\top \xi$,  and $\Xi^i$ is a selection matrix that extracts the first time-horizon component from the control input trajectory.


\section{Stability and Robustness Analysis}
In this section we show that the closed-loop system is locally input-to-state stable (LISS) if we perform enough communications and if the regularization is small enough. In our analysis, we make extensive use of the optimization problem \eqref{eq:centralized-opt2} expressed in the following compact form: 
\begin{subequations} \label{eq:centralized-opt}
\begin{alignat}{2}
&\min_{u}~~&&f(u,x)\\
&~\mathrm{s.t.} && Dx + Cu \leq c \label{eq:local-constraints-condensed}
 \\
& && Eu + Fx\leq b\end{alignat}
\end{subequations}
where $f(u,x) = 1/2 \|(u,x)\|_M^2$ and $M, E, D, C, b$ and $c$ are stacked versions of the matrices in \eqref{eq:centralized-opt2} and Appendix~\ref{ap:POCPM}.

We begin by expressing the closed-loop system in error coordinates, namely,
\begin{subequations} \label{eq:error-system}
\begin{align} 
\Sigma_1:&\begin{cases} \label{eq:error1}
    ~x_{t+1} = A x_t + B(\kappa_\eps(x_t) + \delta u_t) + d_t ,\\
    ~\Delta x_t = h_1(x_t,\delta u_t,d_t)
\end{cases}\\
\Sigma_2: &\begin{cases} \label{eq:error2}
    ~e_{t+1} = \mc G^\ell (e_{t},x_t,\Delta x_t),\\
    ~\delta u_t = h_2(e_t,x_t)
\end{cases}
\end{align}

\end{subequations}
where $e_t = \lambda_t - S^d_\eps(x_t)$, $h_1(x_t,\delta u_t,d_t) = (A-I) x_t + B(\kappa_\epsilon(x_t) + \delta u_t) + d_t$, $h_2(e_t,x_t) = q(e_t + S^d_\eps(x_t),x_t) - q(S^d_\eps(x),x_t)$, $\mc{G}^\ell(e_t,x_t,\Delta x_t) = \mc{T}^\ell(e_t + S^d_\eps(x_t),x_t + \Delta x_t) - S^d_\eps(x_t + \Delta x_t)$, and $\kappa_\eps(x) = q(S^d_\eps(x),x)$.

Next, we will show that both subsystems in \eqref{eq:error-system} are input-to-state stable (ISS) and derive a sufficient condition for ISS of their feedback interconnection using the small-gain theorem \citep{Jiang2004NonlinearApplications}. We begin with the plant.

\begin{thm}(LISS of the Plant Subsystem \eqref{eq:error1}) \label{thm:plant-iss}
There exists $\tilde{\eps}, \gamma_1,\sigma_1 > 0$ and $\beta_1 \in \mc{KL}$ such that if $\eps \leq \tilde{\eps}$ then
\begin{equation}
    \|x_t\| \leq \beta_1(\|x_0\|,t) + \gamma_1~\|\delta u\|_\infty+ \sigma_1~\|d\|_\infty,
\end{equation}
and $\{x_t\} \subseteq \Gamma_N$ for $\|x_0\|, \|\delta u\|_\infty, \|d\|_\infty$ sufficiently small.
\end{thm}
\begin{proof}
See Appendix~\ref{app:plant-iss}.
{\hfill $\square$}
\end{proof}

This theorem shows that the plant is ISS with respect to the exogenous disturbance $d$ as well as the suboptimality error $\delta u$ caused by incomplete optimization if the regularization is small enough. This is expected as a large regularization may warp the MPC feedback law $\kappa_\eps$ enough to destabilize the closed-loop system. Of note is that the regularization does not cause an offset in the equilibrium point, this is because the regularization term drives $\lambda$ toward 0 which matches the unregularized dual solution at $x = 0$ since the constraints must all be inactive at the origin by Assumption~\ref{ass:plant-mpc-properties}.

Next, we show that the algorithm, viewed as a dynamical system, is ISS with respect to the state increment $\Delta x$.
\begin{thm}(ISS of Algorithm subsystem \eqref{eq:error2}) \label{thm:algo-iss}
Let $\eps\! >\! 0$ and let $\alpha \in (0,1/L_{\varphi_\epsilon})$ where $L_{\varphi_\epsilon} $ is the Lipschitz constant of the regularized dual gradient $\nabla \varphi_\epsilon$. If $\ell > \frac{2}{\sqrt{\alpha \eps}} - 1$, then
\begin{equation}
    \|e_t\| \leq \eta(\ell)^t \|e_0\| + \gamma_2(\ell)~\|\Delta x\|_\infty
\end{equation}
where $\eta(\ell) = \frac{2}{\sqrt{\alpha \eps}} \left(\frac{1}{\ell+1}\right) < 1$ and $\gamma_2(\ell) = \frac{\eta(\ell)}{\alpha (1 - \eta(\ell))}$.
\end{thm}
\begin{proof}
See Appendix~\ref{app:algo-iss}.
{\hfill $\square$}
\end{proof}
This result is intuitive in the sense that by Theorem~\ref{thm:algo-conv}, Algorithm~1 converges for any constant $x\in \Gamma_N$, the state increment then acts as a disturbance. Regularization is crucial for this result, as it ensures that the dual solution mapping $S^d_\epsilon$ is well-behaved (single-valued and Lipschitz) and can be readily tracked.

Finally, we show that the interconnection of the two systems is ISS provided that enough algorithm iterations (communication rounds) are performed.
\begin{thm}(ISS of the plant-algorithm interconnection) \label{thm:main-theorem}
Suppose that $\ell$ and $\eps$ are chosen to satisfy the conditions of Theorems~\ref{thm:plant-iss} and \ref{thm:algo-iss}. Then, there exists a finite $\ell^* > 0 $ such that if $\ell \geq \ell^*$ then $\{x_t\} \subseteq \Gamma_N$ and the closed-loop system \eqref{eq:error-system} is LISS with respect to the disturbance $d$.
\end{thm}
\begin{proof}
See Appendix~\ref{app:sys-iss}. {\hfill $\square$}
\end{proof}

This theorem shows that if the communication budget at each sampling period is sufficiently large (enough iteration of the algorithm are performed) then the closed loop system is robustly stable despite the inexact optimization, in fact it recovers the qualitative properties of optimal MPC.
The stability result is necessarily local since the optimal control law is only stabilizing for $\Gamma_N$, the set of states for which the optimization problem has a solution.

\section{Illustrative Example} \label{sec:Numerics}

We consider $M$ robotic agents moving on the 2D plane. Each agent has a state $x^i = (p^i,v^i) \in\mathbb{R}^4$ that encodes position and velocity in the plane $\mathbb{R}^2$, and double-integrator dynamics (with control $u_i \in \mathbb{R}^2$) of the form
\begin{align}
     x^i_{t+1} &= \left(I_2 \otimes 
        \begin{bmatrix}
        1 & 1 \\ 0 & 1
    \end{bmatrix}\right) x^i_t +
        \left( I_2 \otimes 
        \begin{bmatrix}
            0 \\ 1
        \end{bmatrix} \right) u^i_t.
\end{align}
We consider the task of driving the agents to a configuration encoded in the target $\bar x^i = (\bar p^i,0)$ while satisfying input constraints $ u^i \in [\underline{u}^i,~\overline{u}^i]$ for all $i \in \mc M$, and subject to the coupling constraints
\begin{align}
|p^i_t-p^j_t| \leq b_{ij},~~i\neq j, \quad \forall t\in \bb N_{\geq 0}, \label{eq:onenormrelax}
\end{align}
where the absolute value is taken element-wise, and $b_{ij} \in \mathbb{R}^2_{>0}$ are parameters, that ensure that the agents remains within communication distance of each other at all times. This leads to an MPC optimal control problem
\begin{subequations}  \label{eq:agent_sim}
\begin{alignat}{3}
&\min_{u}~~&& \sum_{\agt\in \mc M} \sum_{k \in \mc{H}}\|x^i_k - \bar x^i\|_{Q^i}^2 &&+ \|u^i_k\|_{R^i}^2 \\
&~\mathrm{s.t.} && x^i_{k+1} = A x^i_{k} + B u^i_{k} && \forall i \in \mc{M},~ \forall k \in \mc{H} \\
&~ && x^i_N = \bar x^i,~~u^i_k \in \mc{U}^i && \forall i \in \mc{M},~ \forall k \in \mc{H}  \\
&  && |p^i_k-p^j_k| \leq b_{ij}, &&  \forall i\neq j,~~ \forall k \in \mc{H}\setminus \{0\} \label{eq:couple_con}
\end{alignat}
\end{subequations}
which is of the form \eqref{eq:general_eq_problem} and uses a terminal state constraint (and thus $P^i = 0$ is a valid choice).

We simulate 3 agents throughout a maneuver taking them from an initial formation to a target formation, as depicted in Fig.~\ref{fig:formation}.
The communication graph of the agents is the complete graph on 3 nodes.
We set $Q^i=I_4$ and $R^i=I_2$ for each $i\in\{1,2,3\}$, and each $b_{ij} = \mathbf{1}_2$.
Problem~\eqref{eq:agent_sim} is solved with Algorithm~1 for varying rounds of communication per timestep $l$.
As depicted visually in Fig.~\ref{fig:formation}, the closed-loop dynamics are stable and the agents converge to the target points with even a single round of communication per timestep.

As more rounds of communication per sampling period are performed, the trajectory converges to the optimal one given by the MPC problem. In Fig.~\ref{fig:const_vio}, one can see that increasing the number of rounds of communication per timestep substantially decreases the coupling constraint violations \eqref{eq:couple_con} throughout the maneuver.

\begin{figure}[htbp]
    \centering
    \includegraphics[width=\columnwidth]{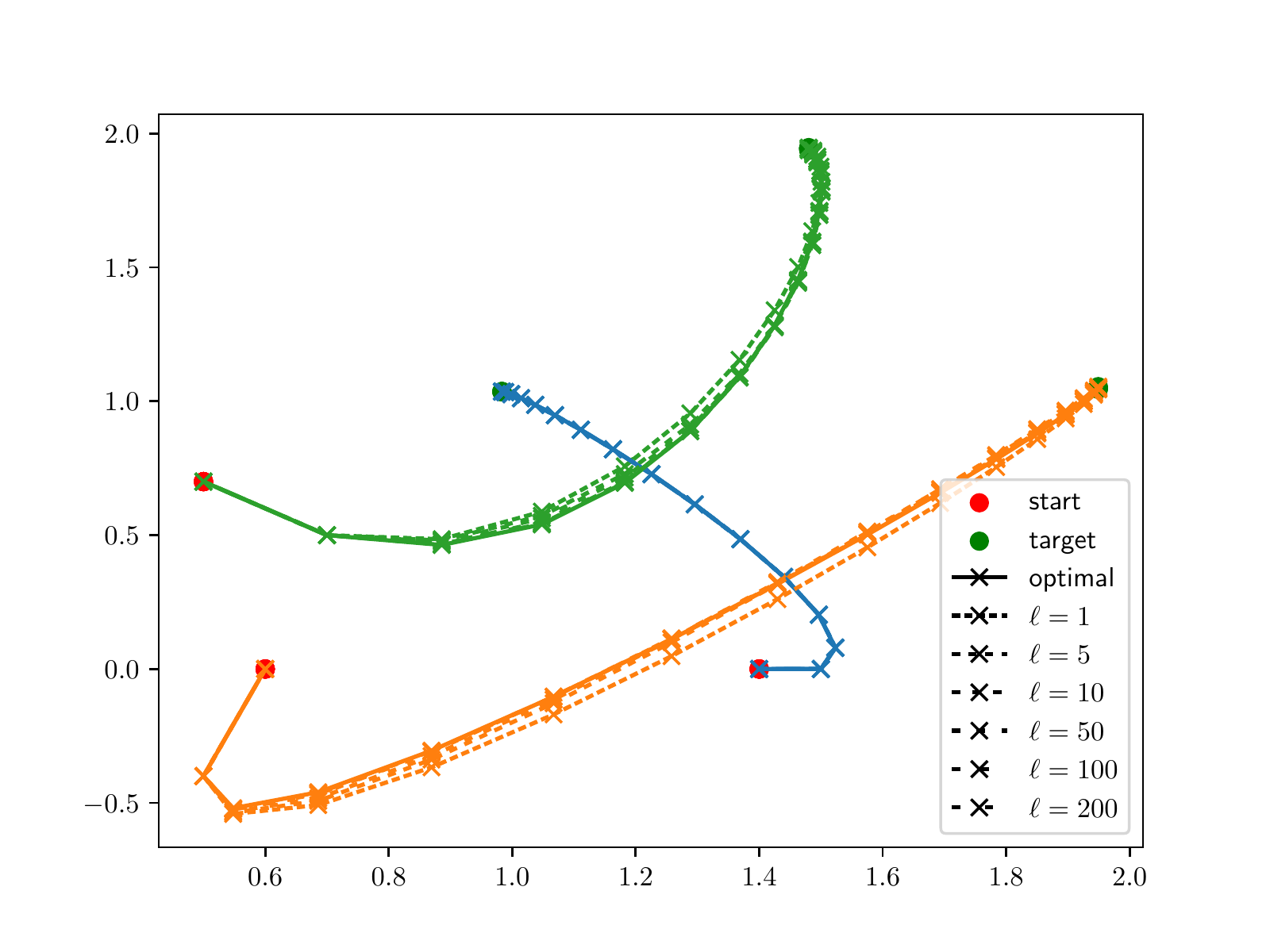}
    \caption{The closed-loop dynamics become stable with a single communication iteration per timestep and converges to the optimal MPC feedback law as more communication resources become available.}
    \label{fig:formation}
\end{figure}

\begin{figure}[htbp]
    \centering
    \includegraphics[width=\columnwidth]{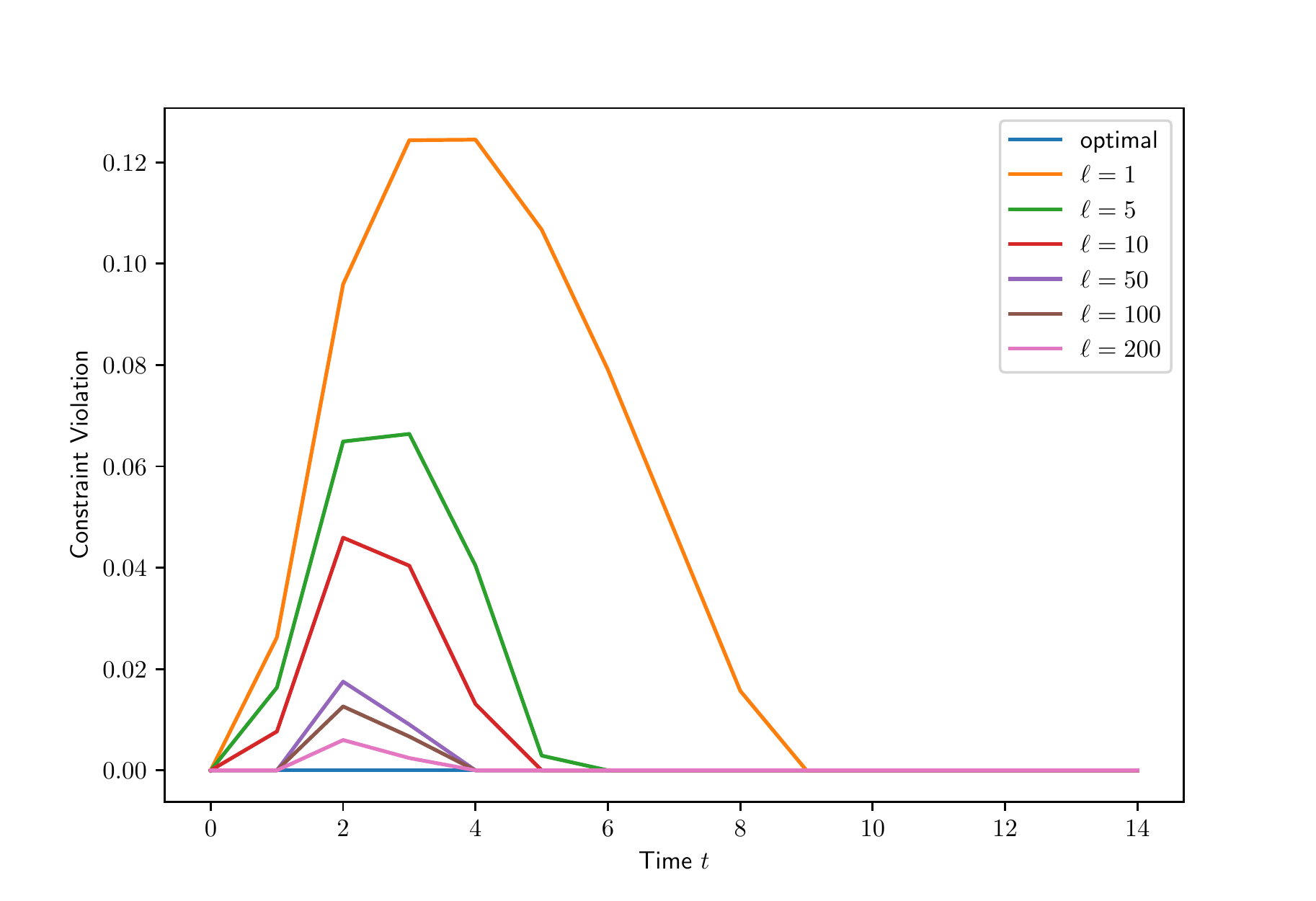}
    \caption{Performing more rounds of communication reduces the constraint violation.}
    \label{fig:const_vio}
\end{figure}

\section{Conclusion} \label{sec:Conclusion}
We proposed an MPC scheme for multi-agent systems in which at each sampling instant the computation for the next control inputs are distributed both in space, across the subsystems, and in time, by maintaining a running solution estimate of the underlying optimal control problem. We proved robust stability of the proposed suboptimal MPC scheme if enough communication between the subsystems are performed at each sampling period. Future research directions include designing a fully-distributed scheme, handling coupling terms in the local cost functions, and precisely quantifying the stability bounds.

\appendix

\section{Condensed POCP Matrices}
\label{ap:POCPM}
The matrices in \eqref{eq:dOptCost} are $H^i = \bar H^i + (I_N \otimes R^i)$, $\bar H^i = (\hat B^i)^\top \hat H^i \hat B^i$, $G^i = (\hat B^i)^\top \hat H^i \hat A^i$, $W^i = Q^i + (\hat A^i)^\top \hat H^i \hat A^i$, $\hat H^i = \textrm{blkdiag}(I_N \otimes Q^i,\, P^i )$,
\begin{align*}
\hat B^i = \begin{bmatrix}
0       & 0       &     0\\
B^i     & 0       &     0\\
\vdots  & \ddots  &\vdots\\
(A^i)^{N-1}B^i& \cdots  &B^i 
\end{bmatrix}, 
 \text{ and }
\hat A^i =
\begin{bmatrix}
I\\
A^i\\
\vdots\\
(A^i)^{N}
\end{bmatrix}.
\end{align*}
Let $\hat L^i = \textrm{blkdiag}(I_N \otimes C_x^i, C_N^i)$, the matrices in \eqref{eq:local-constraints-condensed2} are
\begin{align*}
D^i = \begin{bmatrix}
 0 \\ \hat L^i \hat A^i
\end{bmatrix};\quad
C^i = \begin{bmatrix}
I_N \otimes C_u^i\\[.3em]
\hat L^i \hat B^i
\end{bmatrix}, 
\quad
 c^i =
\begin{bmatrix}
\mathbf 1_N \otimes c_u^i\\
\mathbf 1_{N} \otimes c_x^i\\
c_N^i
\end{bmatrix}.
\end{align*}
Finally, the matrices in \eqref{eq:couplconstr2}  are $b = \mathbf 1_N \otimes \bar b$, $F^i = (I_{N} \otimes E_x^i) \hat A^i_{\text{rblk}[1:N]}$, and  $E^i = ( I_{N} \otimes E_x^i)\hat B^i_{\text{rblk}[1:N]} + (I_N \otimes E_u^i)$, where $\hat A^i_{\text{rblk}[1:N]}$ denotes the first $N$ row blocks of $\hat A^i$, and similarly for $\hat B^i_{\text{rblk}[1:N]}$.

\section{Proof of Theorem~\ref{thm:plant-iss}} \label{app:plant-iss}
We will show that the root of the value function
\begin{equation}
    \phi(x) = \frac{1}{\sqrt{2}} \left \|\begin{bmatrix}
        S^p(x) \\ x
    \end{bmatrix} \right \|_M,~~M = \begin{bmatrix}
        H & G\\
        G^T & W
    \end{bmatrix}
\end{equation}
where $S^p(x) = \argmin{u}\{0.5 \|(u,x)\|_M^2 \; |\; Eu + Fx \leq b, \; Dx + Cu \leq c\}$ is the primal solution mapping, is an ISS Lyapunov function for \eqref{eq:error1}. 

\begin{lem} \label{lmm:value-function}
Given Assumption~\ref{ass:plant-mpc-properties}, there exist $L_p,a_1, a_2 > 0$ and $\beta \in (0,1)$ such that the primal mapping $S^p:\Gamma_N \to \reals^{Nm}$ and the value function $\phi:\Gamma_N \to \reals_{\geq0}$ have the following properties:
\begin{enumerate}[(i)]
\item $\|S^p(x) - S^p(y)\| \leq L_p \|x-y\|$
\item $a_1 \|x\| \leq \phi(x) \leq a_2 \|x\|$
\item $|\phi(x) - \phi(y)| \leq a_2\|x - y\|$
\item $\phi(Ax + B\kappa(x)) \leq \beta \phi(x)$
\end{enumerate}
\end{lem}
\begin{proof}
Point (i) follows from \cite[Theorem 4]{Bemporad2002TheSystems} since \eqref{eq:centralized-opt} is a multi-parametric QP. The lower bound in (ii) follows directly from Assumption~\ref{ass:plant-mpc-properties} since $\phi(x) \geq \sqrt{0.5 (\|x\|_Q^2 + \|\kappa(x)\|_R^2)} \geq \sqrt{0.5} \|Q\| \|x\|$. 

Next we establish Lipschitz continuity, for any $x,y\in \Gamma_N$ we have that
\begin{align*}
|\phi(x) - \phi(y)|^2 &= 1/2|\|(S^p(x),x)\|_M - \|(S^p(y),y)\|_M|^2\\
&\leq 1/2 \|(S^p(x) - S^p(y), x-y)\|_M^2\\
&\leq 1/2 (\|M\| \|S^p(x) - S^p(y)\|^2 + \|x - y\|^2)\\
&\leq 1/2 \|M\| (1+L_p^2) \|x-y\|^2
\end{align*}
where $L_p$ is the Lipschitz constant of $S^p$ and we used the reverse triangle inequality to go from the first to the second line. Thus (iii) holds with $a_2 = \sqrt{0.5 \|M\|(1+L_p^2)}$; the upper bound in (ii) is an immediate corollary.

Since $P$ is chosen as the solution of the discrete-time algebraic Riccati equation, point (iv) follows from \cite[\S 5.6.3]{goodwin2006constrained}. \qed
\end{proof}

Next, we need to characterize the error $w(x) =  \kappa_\eps(x)-\kappa(x)$ in the control input caused by the regularization. 
\begin{lem}
Given Assumption~\ref{ass:plant-mpc-properties}, there exists $L_\kappa > 0$ such that for all $\eps > 0$ and $x\in \Gamma_N$
\begin{equation}
    \|w(x)\| = \|\kappa(x) - \kappa_\eps(x)\| \leq \sqrt{\eps} L_\kappa \|x\|.
\end{equation}
\end{lem}
\begin{proof}

Applying \cite[Prop 3.1]{koshal2011multiuser} to the regularized dual problem for each $x\in \Gamma_N$ yields the bound
\begin{equation} \label{eq:dual_bound}
    \mu \|S^p(x) - S^p_\eps(x)\|^2 + \frac{\eps}{2}\|S^d_\eps(x)\|^2 \leq \frac{\eps}{2}\|\lambda\|^2~~\forall \lambda \in S^d(x)
\end{equation}
where $S^d:\Gamma_N \mto \reals_{\geq 0}^p$ is the dual solution map for the unregularized problem \eqref{eq:DualP}, which may be multi-valued, $S^p_\eps$ is the primal solution mapping associated with the regularized problem \eqref{eq:DualP_reg} such that $\kappa_\eps(x) = \Xi S^p_\eps(x)$, and $\mu > 0$ is the strong convexity constant of $f$ in \eqref{eq:centralized-opt} with respect to $u$. We proceed by finding a Lipschitz like bound for $S^d(x)$. All elements $(u,\lambda) \in (S^p(x),S^d(x))$ satisfy the KKT conditions
\begin{equation}
    \begin{bmatrix}
        H & E^T & C^T\\	
        E & 0 & 0 \\
        C & 0 & 0
    \end{bmatrix} \begin{bmatrix}
        u\\ \lambda \\ \nu
    \end{bmatrix} + \begin{bmatrix}
        Gx \\ Fx-b \\ Dx - c
    \end{bmatrix} + \begin{bmatrix}
        0 \\ \mc{N}_+(\lambda) \\ \mc{N}_+(\nu)
    \end{bmatrix} \ni 0
\end{equation}
of the fully dualized problem where $\nu$ are dual variables associated with the local constraints \eqref{eq:local-constraints-condensed}, where $\mc{N}_+$ denotes the normal cone of the nonnegative orthant with appropriate dimension. This is an affine variational inequality of the form $\tilde Q z + \tilde R x + \tilde b + \mc{N}_C(z) \ni 0$ for matrices $\tilde Q,\tilde R, \tilde b$ and a polyhedral set $C = \reals^{N n_u} \times \reals_{\geq 0}^l$ and thus the solution mapping $\tilde S(x) = (Q + \mc{N}_C)^{-1}(-b-\tilde Rx)$ is polyhedral \cite[Ex.~3D.2]{Dontchev2009ImplicitMappings} and therefore outer Lipschitz continuous, i.e., there exists $L_1 > 0$ such that for all $\bar x \in \Gamma_N$
\begin{equation}
    \tilde S(x) \subseteq \tilde S(\bar x) + L_1 \|x - \bar x\| \mc{B}
\end{equation}
where $\mc{B}$ is the unit ball. By Assumption~\ref{ass:plant-mpc-properties}, we know that at the equilibrium point $x = 0$, $\kappa(x) = 0$ and the state and input pair $(x,\kappa(x))$ must lie in the interior of the constraints. Thus all constraints are inactive at $0$, $\tilde S(0) = \{0\}$ and $\tilde S(x) \subseteq L_1 \|x\| \mc{B}$. This implies that $\|z\| \leq L_1 \|x\|~\forall z\in \tilde S(x)$ and, in particular, that $\|\lambda\| \leq L_1 \|x\|~\forall \lambda \in S^d(x)$. Combining this with \eqref{eq:dual_bound} we obtain 
\begin{align*}
    \|S^p(x) - S^p_\eps(x)\|^2 &\leq \frac{\eps}{2\mu }\|\lambda\|^2 - \frac{\eps}{2\mu }\|S^d_\eps(x)\|^2~~\forall \lambda \in S^d(x)\\
    &\leq \frac{L_1^2}{2\mu }\|x\|^2.
\end{align*}
Since $\kappa_\eps = \Xi S^p_\eps$, taking the square root of both sides and letting $L_\kappa = \|\Xi\| L_1/\sqrt{2\mu}$ completes the proof. \qed
\end{proof}

Finally, we can show $\phi$ is a LISS Lyapunov function. Define $f_w(x,\tilde d) = Ax + B \kappa(x) + Bw(x) + \tilde d$ where $\tilde d = d + B \delta u$ is the combined physical and suboptimality disturbance. Then, using Lemma~\ref{lmm:value-function} and the triangle inequality we have that for all $x\in \Gamma_N$
\begin{align*}
\phi(f_w(x,\tilde d)) &\leq \phi(f_w(x,0)) + |\phi(f_w(x,\tilde d)) - \phi(f_0(x,0))|\\
&\leq \beta \phi(x) + a_2 \|f_w(x,\tilde d) - f_0(x,0)\|\\
& = \beta \phi(x) + a_2 \|w(x) + d + B\delta u\|\\
& \leq \beta \phi(x) + a_2 \|w(x)\| + a_2 \|B\delta u\| + a_2 \|d\|\\
&\leq \left (\beta + L_\kappa \frac{a_2}{a_1} \sqrt{\eps}\right ) \phi(x) + a_2 \|B\delta u\| + a_2 \|d\|,
\end{align*}
and thus since $x_{t+1} = f_w(x_t,\tilde d_t)$
\begin{equation}
    \phi(x_{t+1}) \leq \tilde \beta(\eps) \phi(x_t) + a_2 \|B\| \|\delta u_t\| + a_2 \|d_t\|.
\end{equation}
Note that since $\beta \in (0,1)$, $\tilde \beta(\eps) \in (0,1)$ for $\epsilon < \tilde \eps = \left(\frac{1-\beta}{L_\kappa}\frac{a_1}{a_2}\right)^2$ and thus for small enough $\eps$
we obtain that, following \cite[Example 3.4]{Jiang2001Input-to-stateSystems}, 
\begin{equation*}
    \phi(x_t) \leq \tilde \beta(\eps)^t \phi(x_0) + \frac{a_2 \|B\|}{1- \beta(\eps)} \|\delta u\|_\infty + \frac{a_2}{1- \beta(\eps)} \|d\|_\infty
\end{equation*}
provided $x_t\in \Gamma_N$. Combined with Lemma~\ref{lmm:value-function}.(ii) yields
\begin{equation}\label{eq:ISS-eq}
    \|x_t\| \leq \beta_1(\|x_0\|,t) + \gamma_1~\|\delta u\|_\infty+ \sigma_1~\|d\|_\infty
\end{equation}
with $\beta_1(s,t) = s\tilde \beta(\eps)^ta_2/a_1$, $\sigma_1 = \frac{a_2}{a_1} \frac{1}{1-\tilde \beta(\eps)}$, and $\gamma_1 = \|B\| \sigma_1$. 

The bound \eqref{eq:ISS-eq} requires that $x_t\in \Gamma_N$ for all time, we next demonstrate that if $\|x_0\| \leq \rho/3, \|\delta u\|_\infty \leq \rho/(3\gamma_1)$ and $\|d\|_\infty \leq \rho/(3\sigma_1)$ then $x_t\in \Gamma_N$ for all $t$. Note that by Assumption~\ref{ass:plant-mpc-properties}, $0 \in \mathrm{interior}~\Gamma_N$ and thus by Lemma~\ref{lmm:value-function}.(ii) there exists $\rho > 0$ such that $\mc{D} = \{x~|~\|x\| \leq \rho\} \subset \Gamma_N$. Then we proceed by induction. Clearly $\|x_0\| \leq \rho/2 \leq \rho$; assuming \eqref{eq:ISS-eq} holds up to $t-1$ we have
\begin{align*}
    \|x_t\| &\leq \beta_1(\|x_0\|,t) + \gamma_1~\|\delta u\|_\infty+ \sigma_1~\|d\|_\infty\\
    &\leq \max\{3\|x_0\|, 3\gamma_1~\|\delta u\|_\infty, 3\sigma_1~\|d\|_\infty \}\\
    &= \max\{3 \rho /3, 3\gamma_1\rho / (3\gamma_1), 3\sigma_1  / (3\sigma_1)\rho\}\\
    & = \rho,
\end{align*}
where we have used the identity $a + b + c \leq \max\{3a,3b,3c\}$ for all $a,b,c\geq 0$, and thus $x_t \in \mc{D} \subset \Gamma_N$ as claimed. Note that constraints may be active during transients. \qed

\section{Proof of Theorem~\ref{thm:algo-iss}} \label{app:algo-iss}
We begin by establishing Lipschitz continuity of the regularized primal and dual solution mappings
\begin{lem} \label{lmm:solution-regularity}
The primal and dual solution maps $S^p_\eps$ and $S^d_\eps$ of the regularized problem \eqref{eq:DualP_reg} are $L_p^\eps$ and $L_d^\eps$ Lipschitz continuous on $\Gamma_N$.
\end{lem}
\begin{proof}
The primal-dual solution map $(S^p_\eps,S^d_\eps)$ solves the following strongly monotone variational inequality
\begin{equation} \label{eq:VI-condensed}
    \begin{bmatrix}
        H & E^T \\
        E & \eps I
    \end{bmatrix} \begin{bmatrix}
        u \\ \lambda
    \end{bmatrix} + \begin{bmatrix}
        G x\\
        Fx -b
    \end{bmatrix} + \begin{bmatrix}
        \mc{N}_{\mc{Z}}(u,x)\\
        \mc{N}_{+}(\lambda)
    \end{bmatrix} \ni 0,
\end{equation}
where $\mc{Z} = \bigtimes_{i\in \mc M} \mc{Z}^i$ are the collected local constraints. The variational inequality \eqref{eq:VI-condensed} is necessary and sufficient for optimality of the following multi-parametric QP 
\begin{subequations} \label{eq:pd-reg-qp}
\begin{alignat}{2}
&\min_{u,y}~~&& \frac12  \begin{bmatrix} 
u\\ x \\ y
\end{bmatrix}^T \begin{bmatrix}
    H & G^T & 0\\
    G & W & 0\\
    0 & 0 & \eps^{-1}I
\end{bmatrix} \begin{bmatrix} 
u\\ x \\ y
\end{bmatrix}\\
&~\mathrm{s.t.} && Cu + Dx \leq c\\
& &&Eu + Fx \leq b + y
\end{alignat}
\end{subequations}
as the KKT necessary conditions of \eqref{eq:pd-reg-qp} reduce to \eqref{eq:VI-condensed}, and thus $(S^p_\eps,S^d_\eps)$ is Lipschitz continuous by \cite[Theorem 4]{Bemporad2002TheSystems}. \qed
\end{proof}
By using Theorem~\ref{thm:algo-conv}, we have that
\begin{equation}
    \psi_\eps(S^d_\eps(x),x) - \psi_\eps(T^\ell(\lambda,x),x) \leq \frac{2\alpha^{-1}}{(\ell + 1)^2} \|\lambda - S^d_\eps(x)\|^2,
\end{equation}
while strong concavity of the dual (Lemma~\ref{lmm:dual-props}) yields
\begin{equation}
    \frac{\eps}{2} \|\lambda - S^d_\eps(x)\|^2 \leq \psi_\eps(S^d_\eps(x),x) - \psi_\eps(\lambda,x).
\end{equation}
Combining these two inequalities, we obtain
\begin{equation}
    \|T^\ell(\lambda,x) - S^d_\eps(x)\| \leq \eta(\ell) \|\lambda - S^d_\eps(x)\|,
\end{equation}
where $\eta(\ell)^2 = \frac{4}{\alpha \eps} \left(\frac{1}{\ell+1}\right)^2$. Then we have that
\begin{align*}
\|e_{t+1}\| &= \|\lam_{t+1} - S^d_\eps(x_{t+1})\| \leq \eta(\ell) \|\lam_t - S^d_\eps(x_{t+1})\|\\
& \leq\eta(\ell) \|\lam_t - S^d_\eps(x_t)\| + \eta(\ell)\|S^d_\eps(x_t) - S^d_\eps(x_{t+1})\|\\
&\leq \eta(\ell) \|e_t\| + \eta(\ell) L_d^\eps \|\Delta x_t\|.
\end{align*}
Again following \cite[Example 3.4]{Jiang2001Input-to-stateSystems}, for $\ell > \bar \ell$ it hold that $\eta(\ell) < 1$ which implies that
\begin{equation}
    \|e_t\| \leq \eta(\ell)^t \|e_0\| + \gamma_2(\ell)~\|\Delta x\|_\infty
\end{equation}
where $\gamma_2(\ell) = \frac{\eta(\ell) \alpha^{-1}}{1 - \eta(\ell)}$ as claimed. \qed

\section{Proof of Theorem~\ref{thm:main-theorem}} \label{app:sys-iss}
We begin by noting by Theorem~\ref{thm:plant-iss}, \eqref{eq:error1} is LISS and thus thanks to \cite[Lemma 3.8]{Jiang2001Input-to-stateSystems} we have the asymptotic bound ($\overline{\lim}$ denotes the limit supremum)
\begin{align*}
    \limsupt \|x_t\| &\leq \gamma_1 \limsupt \|\delta u_t\| + \gamma_1 \limsupt\|d_t\|\\
    &\leq \gamma_1 \limsupt \|q(\lam_t,x_t) - q(S^d_\eps(x_t),x_t)\| + \gamma_1 \limsupt\|d_t\|\\
    &\leq L_q \gamma_1 \limsupt \|e_t\| + \gamma_1 \limsupt\|d_t\|,
\end{align*}
where we used that $q(\cdot,x)$ is $L_q$-Lipschitz uniformly in $x$ by strong convexity of $f^i(\cdot,x)$. Moreover by Theorem~\ref{thm:algo-iss}, \eqref{eq:error2} is ISS and thus again using \cite[Lemma 3.8]{Jiang2001Input-to-stateSystems}
\begin{align}
\limsupt \|e_t\| &\leq \gamma_2(\ell) \limsupt\|\Delta x_t\|\\
&\leq \gamma_2(\ell) \left (\limsupt \|x_t\| + \limsupt \|x_{t+1}\|\right)\\
&\leq 2\gamma_2(\ell) \limsupt \|x_t\|.
\end{align}
Together these yield that
\begin{equation}
    \limsupt \|x_t\| \leq 2 L_q \gamma_1 \gamma_2(\ell) \limsupt\|x_t\| + \gamma_1 \limsupt\|d_t\|,
\end{equation}
and thus by the small-gain theorem \cite[Theorem 1]{Jiang2004NonlinearApplications} the interconnection is LISS (and thus respects the LISS restriction $\{x_t\} \subseteq \Gamma_N$ imposed by the plant subsystem) with respect to $d$ if $2 L_q \gamma_1 \gamma_2(\ell) < 1$. Since $\gamma_2(\ell) = \sqrt{\frac{2}{\alpha \eps}} \left(\frac{1}{\ell+1}\right) \to 0$ as $\ell \to \infty$ there exists $\ell^* \in (0,\infty)$ such that $2 L_q \gamma_1 \gamma_2(\ell^*) = 1$ and the small-gain condition is satisfied for all $\ell > \ell^*$. \qed

\balance       
\bibliography{references} 

\begin{thebibliography}{22}
\providecommand{\natexlab}[1]{#1}
\providecommand{\url}[1]{\texttt{#1}}
\providecommand{\urlprefix}{URL }
\expandafter\ifx\csname urlstyle\endcsname\relax
  \providecommand{\doi}[1]{doi:\discretionary{}{}{}#1}\else
  \providecommand{\doi}{doi:\discretionary{}{}{}\begingroup
  \urlstyle{rm}\Url}\fi

\bibitem[{Bauschke et~al.(2011)Bauschke, Combettes et~al.}]{bauschke2011convex}
Bauschke, H.H., Combettes, P.L., et~al. (2011).
\newblock \emph{Convex analysis and monotone operator theory in Hilbert
  spaces}, volume 408.
\newblock Springer.

\bibitem[{Beck and Teboulle(2009)}]{beck2009fast}
Beck, A. and Teboulle, M. (2009).
\newblock A fast iterative shrinkage-thresholding algorithm for linear inverse
  problems.
\newblock \emph{SIAM journal on imaging sciences}, 2(1), 183--202.

\bibitem[{Bemporad et~al.(2002)Bemporad, Morari, Dua, and
  Pistikopoulos}]{Bemporad2002TheSystems}
Bemporad, A., Morari, M., Dua, V., and Pistikopoulos, E.N. (2002).
\newblock {The explicit linear quadratic regulator for constrained systems}.
\newblock \emph{Automatica}, 38(1), 3--20.
\newblock \doi{10.1016/S0005-1098(01)00174-1}.

\bibitem[{Christofides et~al.(2013)Christofides, Scattolini, de~la Pena, and
  Liu}]{christofides2013distributed}
Christofides, P.D., Scattolini, R., de~la Pena, D.M., and Liu, J. (2013).
\newblock Distributed model predictive control: A tutorial review and future
  research directions.
\newblock \emph{Computers \& Chemical Engineering}, 51, 21--41.

\bibitem[{Dontchev and Rockafellar(2009)}]{Dontchev2009ImplicitMappings}
Dontchev, A.L. and Rockafellar, R.T. (2009).
\newblock \emph{Implicit functions and solution mappings}, volume 543.
\newblock Springer.

\bibitem[{Frejo and Camacho(2012)}]{frejo2012global}
Frejo, J.R.D. and Camacho, E.F. (2012).
\newblock Global versus local {MPC} algorithms in freeway traffic control with
  ramp metering and variable speed limits.
\newblock \emph{IEEE Transactions on intelligent transportation systems},
  13(4), 1556--1565.

\bibitem[{Giselsson and Boyd(2015)}]{giselsson2015metric}
Giselsson, P. and Boyd, S. (2015).
\newblock Metric selection in fast dual forward--backward splitting.
\newblock \emph{Automatica}, 62, 1--10.

\bibitem[{Giselsson and Rantzer(2013)}]{giselsson2013feasibility}
Giselsson, P. and Rantzer, A. (2013).
\newblock On feasibility, stability and performance in distributed model
  predictive control.
\newblock \emph{IEEE Transactions on Automatic Control}, 59(4), 1031--1036.

\bibitem[{Goodwin et~al.(2006)Goodwin, Seron, and
  De~Don{\'a}}]{goodwin2006constrained}
Goodwin, G., Seron, M.M., and De~Don{\'a}, J.A. (2006).
\newblock \emph{Constrained control and estimation: an optimisation approach}.
\newblock Springer Science \& Business Media.

\bibitem[{Jiang et~al.(2004)Jiang, Lin, and
  Wang}]{Jiang2004NonlinearApplications}
Jiang, Z.P., Lin, Y., and Wang, Y. (2004).
\newblock {Nonlinear small-gain theorems for discrete-time feedback systems and
  applications}.
\newblock \emph{Automatica}, 40(12), 2129--2136.
\newblock \doi{10.1016/j.automatica.2004.08.002}.

\bibitem[{Jiang and Wang(2001)}]{Jiang2001Input-to-stateSystems}
Jiang, Z.P. and Wang, Y. (2001).
\newblock {Input-to-state stability for discrete-time nonlinear systems}.
\newblock \emph{Automatica}, 37(6), 857--869.
\newblock \doi{10.1016/S0005-1098(01)00028-0}.

\bibitem[{Koshal et~al.(2011)Koshal, Nedi{\'c}, and
  Shanbhag}]{koshal2011multiuser}
Koshal, J., Nedi{\'c}, A., and Shanbhag, U.V. (2011).
\newblock Multiuser optimization: Distributed algorithms and error analysis.
\newblock \emph{SIAM Journal on Optimization}, 21(3), 1046--1081.

\bibitem[{Leung et~al.(2021)Leung, Liao-Mcpherson, and
  Kolmanovsky}]{Leung2021AControl}
Leung, J., Liao-Mcpherson, D., and Kolmanovsky, I.V. (2021).
\newblock {A Computable Plant-Optimizer Region of Attraction Estimate for
  Time-distributed Linear Model Predictive Control}.
\newblock In \emph{Proceedings of the American Control Conference}, volume
  2021-May, 3384--3391. Institute of Electrical and Electronics Engineers Inc.
\newblock \doi{10.23919/ACC50511.2021.9482879}.

\bibitem[{Liao-McPherson et~al.(2020)Liao-McPherson, Nicotra, and
  Kolmanovsky}]{liao2020time}
Liao-McPherson, D., Nicotra, M.M., and Kolmanovsky, I. (2020).
\newblock Time-distributed optimization for real-time model predictive control:
  Stability, robustness, and constraint satisfaction.
\newblock \emph{Automatica}, 117, 108973.

\bibitem[{Liao-McPherson et~al.(2021)Liao-McPherson, Skibik, Leung,
  Kolmanovsky, and Nicotra}]{liao2021analysis}
Liao-McPherson, D., Skibik, T., Leung, J., Kolmanovsky, I., and Nicotra, M.M.
  (2021).
\newblock An analysis of closed-loop stability for linear model predictive
  control based on time-distributed optimization.
\newblock \emph{IEEE Transactions on Automatic Control}, 67(5), 2618--2625.

\bibitem[{Luis et~al.(2020)Luis, Vukosavljev, and Schoellig}]{luis2020online}
Luis, C.E., Vukosavljev, M., and Schoellig, A.P. (2020).
\newblock Online trajectory generation with distributed model predictive
  control for multi-robot motion planning.
\newblock \emph{IEEE Robotics and Automation Letters}, 5(2), 604--611.

\bibitem[{M{\"u}ller and Allg{\"o}wer(2017)}]{muller2017economic}
M{\"u}ller, M.A. and Allg{\"o}wer, F. (2017).
\newblock Economic and distributed model predictive control: Recent
  developments in optimization-based control.
\newblock \emph{SICE Journal of Control, Measurement, and System Integration},
  10(2), 39--52.

\bibitem[{Skibik and Nicotra(2022)}]{skibik2022analysis}
Skibik, T. and Nicotra, M.M. (2022).
\newblock Analysis of time-distributed model predictive control when using a
  regularized primal--dual gradient optimizer.
\newblock \emph{IEEE Control Systems Letters}, 7, 235--240.

\bibitem[{Venkat et~al.(2008)Venkat, Hiskens, Rawlings, and
  Wright}]{venkat2008distributed}
Venkat, A.N., Hiskens, I.A., Rawlings, J.B., and Wright, S.J. (2008).
\newblock Distributed {MPC} strategies with application to power system
  automatic generation control.
\newblock \emph{IEEE transactions on control systems technology}, 16(6),
  1192--1206.

\bibitem[{Yang et~al.(2022)Yang, Wang, Manzie, and Pu}]{yang2022real}
Yang, Y., Wang, Y., Manzie, C., and Pu, Y. (2022).
\newblock Real-time distributed model predictive control with limited
  communication data rates.
\newblock \emph{arXiv preprint arXiv:2208.12531}.

\bibitem[{Zanelli et~al.(2021)Zanelli, Tran-Dinh, and
  Diehl}]{zanelli2021lyapunov}
Zanelli, A., Tran-Dinh, Q., and Diehl, M. (2021).
\newblock A {L}yapunov function for the combined system-optimizer dynamics in
  inexact model predictive control.
\newblock \emph{Automatica}, 134, 109901.

\bibitem[{Zheng et~al.(2016)Zheng, Li, Li, Borrelli, and
  Hedrick}]{zheng2016distributed}
Zheng, Y., Li, S.E., Li, K., Borrelli, F., and Hedrick, J.K. (2016).
\newblock Distributed model predictive control for heterogeneous vehicle
  platoons under unidirectional topologies.
\newblock \emph{IEEE Transactions on Control Systems Technology}, 25(3),
  899--910.

\end{thebibliography}

\end{document}